\documentclass[runningheads]{llncs}
\usepackage{graphicx}
\usepackage[braket, qm]{qcircuit}
\usepackage{url}
\usepackage{hyperref}
\usepackage{xcolor,soul,framed}
\usepackage{amsmath}
\usepackage{amssymb}
\usepackage{bbm}
\usepackage{caption}
\usepackage{subcaption}
\usepackage{microtype}

\definecolor{refblue}{RGB}{102, 102, 153}
\definecolor{stringgreen}{RGB}{80, 107, 65}
\definecolor{keywordblue}{RGB}{64, 89, 245}

\usepackage{listings}
\lstdefinestyle{lststyle}{
  commentstyle=\color{green},
  keywordstyle=\color{keywordblue},
  numberstyle=\tiny\color{gray},
  stringstyle=\color{stringgreen},
  basicstyle=\ttfamily\footnotesize,
  breakatwhitespace=false,
  breaklines=true,
  numbers=none,
  captionpos=b,
  frame=lines,
  keepspaces=true,
  numbers=left,
  numbersep=5pt,
  showspaces=false,
  showstringspaces=false,
  showtabs=false,
  tabsize=2
}
\begin{document}
\title{Uncomputation in the Qrisp high-level Quantum Programming Framework}
\titlerunning{Uncomputation in Qrisp}
%
\author{Raphael Seidel \inst{1}\and
Nikolay Tcholtchev \inst{1}\and
Sebastian Bock \inst{1} \and
Manfred Hauswirth \inst{1,} \inst{2}}
\authorrunning{R. Seidel et al.}
%
\institute{Fraunhofer Institute for Open Communications Systems, Kaiserin-Augusta-Allee 31, 10589 Berlin, Germany \\ \email{firstname.lastname@fokus.fraunhofer.de}\and
TU Berlin, Straße des 17. Juni 135, 10623 Berlin, Germany
\email{firstname.lastname@tu-berlin.de}}
\maketitle              
\begin{abstract}
Uncomputation is an essential part of reversible computing and plays a vital role in quantum computing. Using this technique, memory resources can be safely deallocated without performing a non-reversible deletion process. For the case of quantum computing, several algorithms depend on this as they require disentangled states in the course of their execution. Thus, uncomputation is not only about resource management, but is also required from an algorithmic point of view. However, synthesizing uncomputation circuits is tedious and can be automated. In this paper, we describe the interface for automated generation of uncomputation circuits in our Qrisp framework. Our algorithm for synthesizing uncomputation circuits in Qrisp is based on an improved version of ``Unqomp'',  a solution presented by Paradis et. al. Our paper also presents some improvements to the original algorithm, in order to make it suitable for the needs of a high-level programming framework. Qrisp itself is a fully compilable, high-level programming language/framework for gate-based quantum computers, which abstracts from many of the underlying hardware details. Qrisp's goal is to support a high-level programming paradigm as known from classical software development. 

\end{abstract}
\section{Introduction}
While the hardware side of quantum computing has seen steady improvements, significant progress in quantum software development methods is still lacking. This is due to the fact that coding algorithms for the main available physical backends is still done using quantum circuit objects, which are indeed expressive but provide little structure. In order to better support more complex algorithms, which might include a multitude of concepts, a more abstract programming workflow is necessary.

This problem has been identified by the community and two solutions have been proposed: Q\#~\cite{q_sharp} and Silq~\cite{silq}. Unfortunately, these proposals currently provide no straightforward way of compiling their algorithms into quantum circuits. In previous work on Qrisp~\cite{Qrisp}, we demonstrated several constructs and abstractions, which permit a high-level programming workflow, while still maintaining full platform-independent compilability.
The fundamental paradigm behind Qrisp's design has always been the automation of as many of the repetitive steps of low-level programming as possible without losing expressiveness. As uncomputation is a central and re-occurring topic in many quantum algorithms, it is natural to investigate the automation of this procedure and how Qrisp can support it. In the following, we present an interface for automatic uncomputation within Qrisp as well as some adjustments to the underlying algorithm ``Unqomp''~\cite{unqomp}.

The rest of this paper is organized as follows: Section~\ref{Brief_Overview_of_Qrisp} overviews our Qrisp framework for high-level programming of quantum computers. Then Section~\ref{The_Need_for_Uncomputation_in_Quantum_Computing} motivates the role of uncomputation for quantum software development. Sections~\ref{The_Challenge_of_Implementing_Uncomputation} and \ref{Utilizing_the_Unqomp_Method_in_Qrisp} discuss the possible methods for implementing uncomputation and present the corresponding Qrisp interface. After that, in Section~\ref{Improving_Unqomp_for_the_Needs of_Qrisp} the improvements to established uncomputation methods, which make those more comfortable to use in the scope of Qrisp, are discussed. The final section summarizes and concludes our paper.

\section{Brief Overview of Qrisp}\label{Brief_Overview_of_Qrisp}

The state of the art in programming a quantum computer is currently similar to programming in assembler on a classical computer. Even worse, while assembly programming offers at least some basic instructions, e.g. commands, registers, loops, etc., which are more abstract than accessing the actual hardware gates through binary codes, in quantum computing gates and qubits is the current standard way of programming. Frameworks such as Qiskit~\cite{qiskit} or Cirq~\cite{Heim2020} enable the user to create sub-circuits that can be reused in larger, more complex circuits. However, the handling of the circuits is still quite complicated and tedious. 

The Qrisp framework~\cite{Qrisp_EU} consists of a set of Python modules and language extensions that attempt to overcome the above challenge by abstracting the qubit and gate structure of the underlying circuits as far as possible. This is achieved by conceptually replacing gates and qubits with functions and variables. In this way, it is possible to create much more complex circuits than would be possible with the current low-level approach. It goes without saying that the transition to variables and functions does not mean the end of programming with gates and qubits. The elementary quantum functions must of course still be implemented in the background with the help of gates and qubits. 

\section{The Need for Uncomputation in Quantum Computing}\label{The_Need_for_Uncomputation_in_Quantum_Computing}

Uncomputation is an important aspect of quantum information processing (and reversible computing in general), because it facilitates the efficient use of quantum resources. In classical computing, resource efficiency can be achieved by deleting information from main memory and reusing the deleted bits for other purposes.\footnote{A good correspondence in classical computing to uncomputation in quantum computing is the concept of garbage collection as, e.g., in Java. While classical garbage collection usually simply performs a non-reversible deletion of the collected data, uncomputation in contrast means performing the necessary (reversible) steps to bring the data back into some initial state.} Deleting or resetting a qubit, however, is not a reversible process and is usually performed by measuring the qubit in question and performing a bit flip based on the outcome. This measurement collapses the superposition of other entangled qubits, which are supposed to be unaffected. In many cases this collapse interferes with the quantum algorithm, such that the resulting state can no longer be used.

In some situations, uncomputation is not only relevant as a way to manage quantum resources but is actually required, in order for a quantum algorithm to produce a correct result. One such example is Grover's algorithm~\cite{grover}, which utilizes an oracle function and a diffusing process, in order to search for suitable solutions in a given space based on the logic of the above mentioned oracle function.  Assume that we have two quantum variables, of which one is in a state of uniform superposition:

\begin{align}
\ket{\psi_0} = \sum_{i = 0}^{2^n-1} \ket{i} \ket{0}.
\end{align}

In Grover's algorithm, an oracle now calculates a boolean value $f(i)$, which is required to perform a phase tag on the correct solution, for which the algorithm is searching:

\begin{align}
\ket{\psi_1} = \sum_{i = 0}^{2^n-1} \ket{i} \ket{f(i)}.
\end{align}

After performing the phase tag, the state is:

\begin{align}
\label{oracle_tag}
\ket{\psi_2} = Z_{\ket{f(i)}} \ket{\psi_1}
= \sum_{i = 0}^{2^n-1} (-1)^{f(i)} \ket{i} \ket{f(i)}.
\end{align}

In order for Grover's diffuser to actually amplify the amplitude of the tagged state, we need the state to be disentangled, i.e.,

\begin{align}
\ket{\psi_3} = \sum_{i = 0}^{2^n-1} (-1)^{f(i)} \ket{i} \ket{0}
\end{align}
   
Therefore we need to uncompute the variable containing $f(i)$. This shows clearly why uncomputation is required even
within the implementation and execution of one of the most popular quantum algorithms, in order
to enable the efficient search of solutions for a particular problem based on an oracle function.

\section{The Challenge of Implementing Uncomputation}\label{The_Challenge_of_Implementing_Uncomputation}

In many cases, uncomputing a set of qubits can be achieved by applying the inverse of the steps required for the computation. While this seems like a simple approach, it can be ambiguous, e.g., because it might not be clear which gates actually contributed to the computation. For instance, consider the Z gate in  eq.~\ref{oracle_tag}. In any case, it is a tedious amount of extra programming work, which should be automated.

To remedy this problem, an algorithm called ``Unqomp'' for automatic uncomputation has been devised \cite{unqomp} .
An important advantage of Unqomp is, that it does not follow the philosophy of simply reverting the computation, but rather enables the algorithm to skip  ``un-uncomputation'' or recomputation. Recomputation is a phenomenon that happens if one naively reverts the computation process. If that computation process contained an uncomputation itself, the reversed process contains a recomputation. Unqomp enables skipping that recomputation, by inserting the reverted operations at the correct point within the circuit instead of appending them at the end. While skipping recomputation is generally a truly useful feature, it also has a drawback: Due to the insertion of the reverted operations within the circuit, the qubits holding the values that would potentially be recomputed cannot be deallocated until the uncomputation is completed. If we recompute them, these qubits can be used for other purposes between their un- and recomputation.

In Qrisp the developer can choose\footnote{An algorithm for automatic determination wether a variable should be recomputed has been presented in \cite{Meuli2019}. This method is however not yet implemented within Qrisp} whether they want to perform recomputation: Using the \verb|gate_wrap| decorator, functions of quantum variables can be packed into self-contained gate objects, which are not dissolved by the Unqomp implementation. Any uncomputed quantum variables inside these objects will be recomputed if required. The advantages and drawbacks of uncomputation with and without recomputation are summarized in Fig.~\ref{uncomputation_strategies}.

\pagebreak[4]

\begin{figure}[htb]
    \begin{subfigure}[t]{0.5\textwidth}
\scalebox{0.9}{
\Qcircuit @C=1.0em @R=0.8em @!R { \\
	 	\nghost{{qb_0} :  } & \lstick{{qb_0} :  } & \ctrl{1} & \qw & \ctrl{1} & \qw & \ctrl{1} & \qw & \ctrl{1} & \qw & \qw\\
	 	\nghost{{qb_1} :  } & \lstick{{qb_1} :  } & \ctrl{1} & \qw & \ctrl{1} & \qw & \ctrl{1} & \qw & \ctrl{1} & \qw & \qw\\
	 	\nghost{{qb_2} :  } & \lstick{{qb_2} :  } & \targ & \ctrl{1} & \targ & \gate{\mathrm{V}} & \targ & \ctrl{1} & \targ & \qw & \qw\\
	 	\nghost{{qb_3} :  } & \lstick{{qb_3} :  } & \qw & \ctrl{1} & \qw & \qw & \qw & \ctrl{1} & \qw & \qw & \qw\\
	 	\nghost{{qb_4} :  } & \lstick{{qb_4} :  } & \qw & \targ & \gate{\mathrm{U}} & \qw & \qw & \targ & \qw & \qw & \qw\\
\\ }}
         \caption{\label{recomputation} Conceptual visualisation of an uncomputation, which simply reverts the computation. A value is computed into qubit 4 using some temporary result in qubit 2. Subsequently some process $U$ is applied and afterwards, the computation is reversed. Note that qubit 2 is available for some other process $V$ while $U$ is running. Within Qrisp, automatic qubit allocation permits that the recomputation of qubit 2 can even happen on a differing qubit, enabling automated and flexible resource management.}
     \end{subfigure}
     \ \ \ \ \ 
     \begin{subfigure}[t]{0.5\textwidth}
\scalebox{0.9}{
\Qcircuit @C=1.0em @R=0.8em @!R { \\
	 	\nghost{{qb_0} :  } & \lstick{{qb_0} :  } & \ctrl{1} & \qw & \qw & \qw & \ctrl{1} & \qw & \qw & \qw\\
	 	\nghost{{qb_1} :  } & \lstick{{qb_1} :  } & \ctrl{1} & \qw & \qw & \qw & \ctrl{1} & \qw & \qw & \qw\\
	 	\nghost{{qb_2} :  } & \lstick{{qb_2} :  } & \targ & \ctrl{1} & \qw & \ctrl{1} & \targ & \gate{\mathrm{V}} & \qw & \qw\\
	 	\nghost{{qb_3} :  } & \lstick{{qb_3} :  } & \qw & \ctrl{1} & \qw & \ctrl{1} & \qw & \qw & \qw & \qw\\
	 	\nghost{{qb_4} :  } & \lstick{{qb_4} :  } & \qw & \targ & \gate{\mathrm{U}} & \targ & \qw & \qw & \qw & \qw\\
\\ }}
         \caption{Conceptual visualisation of an uncomputation as done in Unqomp. Note, that only 4 instead of 6 Toffoli gates are needed, because qubit 2 does not need to be recomputed. However, this comes at the cost that qubit 2 is not uncomputed until qubit 4 is uncomputed. Therefore the process $V$ needs to wait until after the uncomputation is done. If the duration of the (re)computation is small compared to the duration of $U$ and $V$, this implementation requires twice the time compared to the approach in \ref{recomputation}.}
         \label{unqomp_strategy}
         \end{subfigure} \caption{\label{uncomputation_strategies} Conceptual visualisation of different uncomputation strategies.}
\end{figure}

\section{Utilizing the Unqomp Method in Qrisp}\label{Utilizing_the_Unqomp_Method_in_Qrisp}

Unqomp has been implemented in Qrisp and we provide two ways to call this function as described in the following subsections.

\subsection{Decorator based Uncomputation in Qrisp}
The first option is the \verb|auto_uncompute| decorator, which automatically uncomputes all local quantum variables, i.e., \verb|QuantumVariable| class instances, of a function.
To demonstrate this functionality, we create a function which returns a \verb|QuantumBool| instance in Qrisp containing the AND value of the three associated inputs. To do so, this function creates a local QuantumBool, which stores the temporary result of the AND value of the first two inputs.

\medskip
\begin{lstlisting}[language = Python, numbers = none]
from qrisp import QuantumBool, mcx

def triple_AND(a, b, c):

  local = QuantumBool()
  result =  QuantumBool()
  
  mcx([a, b], local_quantum_bool)
  mcx([local_quantum_bool, c], result)
  
  return result

a = QuantumBool()
b = QuantumBool()
c = QuantumBool()

result = triple_AND(a, b, c)
\end{lstlisting}
\medskip

Executing this piece of code and visualizing the \verb|.qs| attribute (the \verb|QuantumSession|\footnote{In Qrisp, a \texttt{QuantumSession} contains all the high-level objects and steer the interaction with the hardware or simulation backend.}) of any of the participating \verb|QuantumVariable|s produces the following circuit:
\begin{center}
\scalebox{1.0}{
\Qcircuit @C=1.0em @R=0.8em @!R { 
	 	\nghost{{a.0} :  } & \lstick{{a.0} :  } & \ctrl{1} & \qw & \qw & \qw\\
	 	\nghost{{b.0} :  } & \lstick{{b.0} :  } & \ctrl{1} & \qw & \qw & \qw\\
	 	\nghost{{local.0} :  } & \lstick{{local.0} :  } & \targ & \ctrl{1} & \qw & \qw\\
	 	\nghost{{c.0} :  } & \lstick{{c.0} :  } & \qw & \ctrl{1} & \qw & \qw\\
	 	\nghost{{result.0} :  } & \lstick{{result.0} :  } & \qw & \targ & \qw & \qw\\
}}
\end{center}
We see that the qubit containing the local \verb|QuantumBool| does not end up in the $\ket{0}$ state, if \verb|a| and \verb|b| are in the $\ket{1}$ state. Therefore this qubit is still entangled and cannot be reused for other purposes.

We will now rewrite this function with the \verb|auto_uncompute| decorator:

\medskip
\begin{lstlisting}[language = Python, numbers = none]
from qrisp import QuantumBool, mcx, auto_uncompute

@auto_uncompute
def triple_AND(a, b, c):

  local = QuantumBool()
  result =  QuantumBool()
  
  mcx([a, b], local_quantum_bool)
  mcx([local_quantum_bool, c], result)
  
  return result

a = QuantumBool()
b = QuantumBool()
c = QuantumBool()

result = triple_AND(a, b, c)
\end{lstlisting}
\medskip

This snippet produces the following \verb|QuantumCircuit|:

\begin{center}
\scalebox{1.0}{
\Qcircuit @C=1.0em @R=0.8em @!R {
	 	\nghost{{a.0} :  } & \lstick{{a.0} :  } & \multigate{3}{\mathrm{pt2cx}}_<<<{0} & \qw & \multigate{3}{\mathrm{pt2cx}}_<<<{0} & \qw & \qw\\
	 	\nghost{{b.0} :  } & \lstick{{b.0} :  } & \ghost{\mathrm{pt2cx}}_<<<{1} & \qw & \ghost{\mathrm{pt2cx}}_<<<{1} & \qw & \qw\\
	 	\nghost{{c.0} :  } & \lstick{{c.0} :  } & \ghost{\mathrm{pt2cx}} & \ctrl{1} & \ghost{\mathrm{pt2cx}} & \qw & \qw\\
	 	\nghost{{local.0} :  } & \lstick{{local.0} :  } & \ghost{\mathrm{pt2cx}}_<<<{2} & \ctrl{1} & \ghost{\mathrm{pt2cx}}_<<<{2} & \qw & \qw\\
	 	\nghost{{result.0} :  } & \lstick{{result.0} :  } & \qw & \targ & \qw & \qw & \qw\\
}}
\end{center}

As an effect of the \verb|auto_uncompute| decorator, we see that the multi-controlled X-gate acting on the local QuantumBool has been replaced by a gate called \textit{pt2cx}, which stands for \textit{phase tolerant two controlled X} gate. Phase tolerant logic synthesis is an efficient way of encoding boolean functions into quantum circuits \cite{Seidel_2023}. Using this way of synthesizing boolean functions, we significantly reduce the required resources at the cost of producing an extra phase, depending on the input constellation. However, this phase is reverted once the inverted gate is performed on the same input. For the case of two controls, this is implemented as the so-called Margolus gate~\cite{margolus}. Hence, we can observe how the usage of the \verb|auto_uncompute| decorator in Qrisp can modify the underlying quantum circuits in a way, such that resources, i.e., qubits, that are not needed can be freed and disentangled when required.

\subsection{Uncomputation over Qrisp QuantumVariables}

The second way of invoking uncomputation is the \verb|.uncompute| method of the \verb|QuantumVariable| class in Qrisp. 
A \verb|QuantumVariable| in Qrisp is an abstraction that allows the user to manage several qubits simultaneously and to implement data types such as \verb|QuantumFloat|, \verb|QuantumChar| or \verb|QuantumString|.

We demonstrate the use of the \verb|uncompute| method of the \verb|QuantumVariable| class based on the example from above:

\medskip
\begin{lstlisting}[language = Python, numbers = none]
def triple_AND(a, b, c):

  local = QuantumBool()
  result =  QuantumBool()
  
  mcx([a, b], local_quantum_bool)
  mcx([local_quantum_bool, c], result)
  
  local_quantum_bool.uncompute()
  
  return result

a = QuantumBool()
b = QuantumBool()
c = QuantumBool()

result = triple_AND(a, b, c)
\end{lstlisting}
\medskip

This produces the following quantum circuit:
\begin{center}
\scalebox{1.0}{
\Qcircuit @C=1.0em @R=0.8em @!R {
	 	\nghost{{a.0} :  } & \lstick{{a.0} :  } & \multigate{2}{\mathrm{pt2cx}}_<<<{0} & \qw & \multigate{2}{\mathrm{pt2cx}}_<<<{0} & \qw & \qw\\
	 	\nghost{{b.0} :  } & \lstick{{b.0} :  } & \ghost{\mathrm{pt2cx}}_<<<{1} & \qw & \ghost{\mathrm{pt2cx}}_<<<{1} & \qw & \qw\\
	 	\nghost{{local.0} :  } & \lstick{{local.0} :  } & \ghost{\mathrm{pt2cx}}_<<<{2} & \ctrl{1} & \ghost{\mathrm{pt2cx}}_<<<{2} & \qw & \qw\\
	 	\nghost{{c.0} :  } & \lstick{{c.0} :  } & \qw & \ctrl{1} & \qw & \qw & \qw\\
	 	\nghost{{result.0} :  } & \lstick{{result.0} :  } & \qw & \targ & \qw & \qw & \qw\\
}}
    
\end{center}

The  \verb|uncompute| method and the \verb|auto_uncompute| decorator automatically call the \verb|delete| method after successful uncomputation, which frees the used qubit. If we allocate a new \verb|QuantumBool|, the compiled quantum circuit will reuse that qubit:

\medskip\begin{lstlisting}[language = Python, numbers = none]
from qrisp import cx
d = QuantumBool()
cx(result, d)
\end{lstlisting}
\medskip

And the quantum circuit is updated to:

\begin{center}
\scalebox{1.0}{
\Qcircuit @C=1.0em @R=0.8em @!R {
	 	\nghost{{a.0} :  } & \lstick{{a.0} :  } & \multigate{2}{\mathrm{pt2cx}}_<<<{0} & \qw & \multigate{2}{\mathrm{pt2cx}}_<<<{0} & \qw & \qw & \qw\\
	 	\nghost{{b.0} :  } & \lstick{{b.0} :  } & \ghost{\mathrm{pt2cx}}_<<<{1} & \qw & \ghost{\mathrm{pt2cx}}_<<<{1} & \qw & \qw & \qw\\
	 	\nghost{{d.0} :  } & \lstick{{d.0} :  } & \ghost{\mathrm{pt2cx}}_<<<{2} & \ctrl{1} & \ghost{\mathrm{pt2cx}}_<<<{2} & \targ & \qw & \qw\\
	 	\nghost{{c.0} :  } & \lstick{{c.0} :  } & \qw & \ctrl{1} & \qw & \qw & \qw & \qw\\
	 	\nghost{{result.0} :  } & \lstick{{result.0} :  } & \qw & \targ & \qw & \ctrl{-2} & \qw & \qw\\
}}
\end{center}

We can see how the qubit holding the local \verb|QuantumBool| has been reused to now accomodate the \verb|QuantumBool d|.

\newcommand{\mycomment}[1]{}

\mycomment{
\subsection{Uncomputing multi-qubit gates}
In some cases, the entanglement structure of a set of qubits only permits uncomputation if all of them are uncomputed jointly. In this situation, setting the keyword argument \verb|do_it| to False marks a \verb|QuantumVariable| for uncomputation but does not actually perform it. On the next call with \verb|do_it = True|, the whole batch is uncomputed jointly. The following example illustrates this behaviour:

\medskip
\begin{lstlisting}[language = Python, numbers = none]
from qrisp import gate_wrap

@gate_wrap
def fanout(a, b, c):
   cx(a,b)
   cx(a,c)

a = QuantumBool()
b = QuantumBool()
c = QuantumBool()

fanout(a,b,c)
\end{lstlisting}
\medskip

The \verb|gate_wrap| decorator bundles the quantum instructions inside that function into a single gate:

\begin{center}

\scalebox{1.0}{
\Qcircuit @C=1.0em @R=1.0em @!R {
    \nghost{{a.0} :  } & \lstick{{a.0} :  } & \multigate{2}{\mathrm{fanout}}_<<<{0} & \qw & \qw\\
    \nghost{{b.0} :  } & \lstick{{b.0} :  } & \ghost{\mathrm{fanout}}_<<<{1} & \qw & \qw\\
    \nghost{{c.0} :  } & \lstick{{c.0} :  } & \ghost{\mathrm{fanout}}_<<<{2} & \qw & \qw\\
}}
\end{center}

Trying to call \verb|b.uncompute()| results in an error, because \verb|c| would also be uncomputed if the fanout gate was reverted. Obviously the same would happen to \verb|c.uncompute()|. We overcome this problem by first queuing one of these \verb|QuantumVariable|s for uncomputation without actually performing it:

\medskip
\begin{lstlisting}[language = Python, numbers = none]
b.uncompute(do_it = False)
c.uncompute()
\end{lstlisting}
\medskip

This generates the following circuit:
\begin{center}
\scalebox{1.0}{
\Qcircuit @C=1.0em @R=1.0em @!R {
	 	\nghost{{a.0} :  } & \lstick{{a.0} :  } & \multigate{2}{\mathrm{fanout}}_<<<{0} & \multigate{2}{\mathrm{fanout\_dg}}_<<<{0} & \qw & \qw\\
	 	\nghost{{b.0} :  } & \lstick{{b.0} :  } & \ghost{\mathrm{fanout}}_<<<{1} & \ghost{\mathrm{fanout\_dg}}_<<<{1} & \qw & \qw\\
	 	\nghost{{c.0} :  } & \lstick{{c.0} :  } & \ghost{\mathrm{fanout}}_<<<{2} & \ghost{\mathrm{fanout\_dg}}_<<<{2} & \qw & \qw\\
}}
\end{center}

The \textit{dg} in \textit{fanout\_dg} stands for \textit{dagger} and indicates that this gate is the inverse of \textit{fanout}.
This problem might seem a bit constructed, because the \textit{fanout} gate could in principle be decomposed into a sequence of CNOT gates, which would face no such issue. Avoiding to decompose gates during uncomputation allows the user to deploy the previously mentioned recomputation (if required). An additional feature relatung to this will be highlighted in Section~\ref{synth_gates}.}

\subsection{Case Study: Solving Quadratic Equations using Grover's Algorithm}

We want to close this section with an example given in our previous article \cite{Qrisp}, where we highlighted how Qrisp can be used to solve a quadratic equation using Grover's algorithm. To achieve this in Qrisp,  we employed manual uncomputation using the \verb|invert| environment. Using the \verb|auto_uncompute| decorator we reformulate the code from \cite{Qrisp} as follows:
\newpage
\medskip
\begin{lstlisting}[language = Python, numbers = none]
from qrisp import QuantumFloat, h, z, auto_uncompute

@auto_uncompute
def sqrt_oracle(qf):
    z(qf*qf == 0.25)

qf = QuantumFloat(3, -1, signed = True)
n = qf.size
iterations = int((n/2)**0.5) + 1
h(qf)

from qrisp.grover import diffuser
for i in range(iterations):
    sqrt_oracle(qf)
    diffuser(qf)

result = qf.get_measurement(plot = True)
\end{lstlisting}
\medskip

The function \verb|sqrt_oracle| applies a Z gate onto the \verb|QuantumBool| generated by evaluating the comparison. Note that this \verb|QuantumBool| and the result of the multiplication \verb|qf*qf| (a \verb|QuantumFloat|) are uncomputed automatically.

The histogram of the simulated outcome probabilities is shown in Figure~\ref{fig:grover_plot}, demonstrating
the correctness of the quadratic equation solving procedure, while in parallel uncomputing qubits using the Qrisp
infrastructure for the efficient execution of the Grover's algorithm.

\begin{figure}
    \centering
    \includegraphics[scale = 0.5]{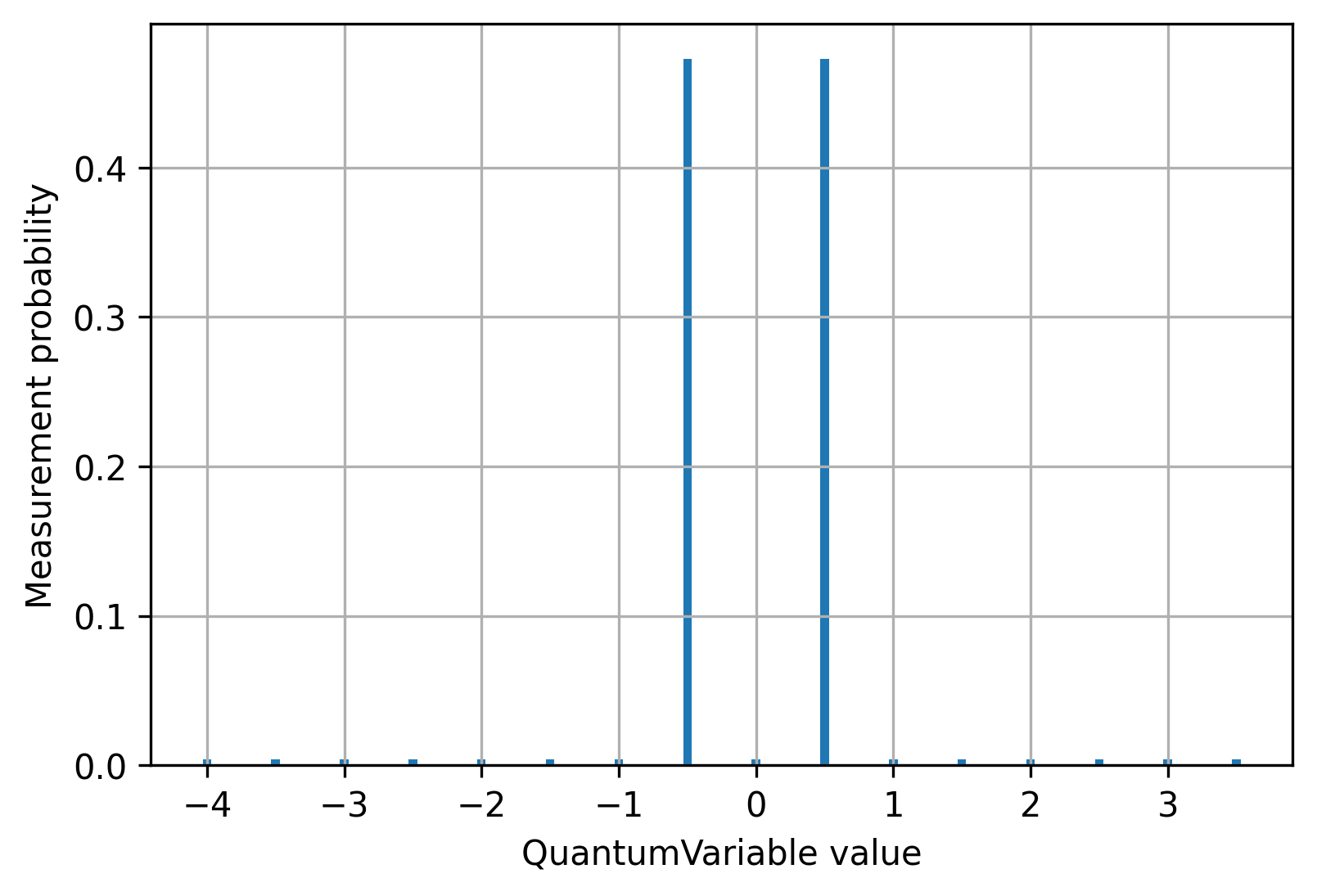}
    \caption{Histogram of the simulation results of the quadratic solver.}
    \label{fig:grover_plot}
\end{figure}

\section{Uncomputation beyond Unqomp}\label{Improving_Unqomp_for_the_Needs of_Qrisp}

Even though the Unqomp algorithm provides a very convenient way of solving automatic uncomputation, it comes with a few restrictions. We will not detail them too deeply here as they are well documented in the original publication~\cite{unqomp}, however, the most important one can be overcome using the Qrisp implementation of the algorithm which is described in the following section.

\subsection{Uncomputing Synthesized Gates}
\label{synth_gates}
The main restriction Unqomp imposes is that only a certain class of gates can be uncomputed, which the authors of Unqomp call \textit{qfree}: A quantum gate is \textit{qfree} if it neither introduces nor destroys states of superposition. In more mathematical terms, this implies that the unitary matrix of a \textit{qfree} gate can only have a single non-zero entry per column.

This is a serious restriction, since many quantum functions make use of non-\textit{qfree} gates such as the Hadamard, even though their net-effect is \textit{qfree}. An example of such a situation is Fourier arithmetic (of which Qrisp's arithmetic module makes heavy use). Even though the multiplication function
\begin{align}
U_{mul}\ket{a}\ket{b}\ket{0} = \ket{a}\ket{b}\ket{a \cdot b}
\end{align}
itself is qfree, it makes use of Hadamard gates, which are not qfree.
In order to overcome this major restriction, the Qrisp implementation of Unqomp does not decompose gates but instead check the combined gate for \textit{qfree}-ness.

This feature, in combination with the previously mentioned \verb|gate_wrap| decorator, can be used to create quantum functions that can be successfully uncomputed even though their inner workings contain non-\textit{qfree} gates.

We demonstrate this with an implementation of the Margolus gate taken from \cite{Maslov_2016}

\medskip
\begin{lstlisting}[language = Python, numbers = none]
from qrisp import cx, ry
from numpy import pi

def margolus(control, target):
    ry(pi/4, target[0])
    cx(control[1], target[0])
    ry(-pi/4, target[0])
    cx(control[0], target[0])
    ry(pi/4, target[0])
    cx(control[1], target[0])
    ry(-pi/4, target[0])
\end{lstlisting}
\medskip

While the Margolus gate itself is \textit{qfree}, the constituents (to be specific, the RY gates) are not. Therefore the following code results in an error:

\medskip
\begin{lstlisting}[language = Python, numbers = none]
from qrisp import QuantumVariable

control = QuantumVariable(2)
target = QuantumVariable(1)

margolus(control, target)

target.uncompute()
\end{lstlisting}
\medskip

We can circumvent this error by applying the \verb|gate_wrap| decorator to \verb|margolus|.

\medskip
\begin{lstlisting}[language = Python, numbers = none]
from qrisp import gate_wrap
control = QuantumVariable(2)
target = QuantumVariable(1)

margolus_wrapped = gate_wrap(margolus)
margolus_wrapped(control, target)

target.uncompute()
\end{lstlisting}
\medskip

Qrisp now automatically checks the combined gate for \textit{qfree}-ness instead of the constituents. Since \textit{qfree}-ness corresponds to the unitary having only a single non-zero entry per column, this property can be checked in linear time if the unitary is known.

\subsection{Permeability}

\textit{Permeability} is a concept, which is introduced in the Qrisp implementation of Unqomp and generalizes the notion of a controlled operation. The permeability status of a gate on a certain input qubit $q$ decides how this gate is treated, when $q$ is uncomputed. We choose this definition, because it permits a broader scope of uncomputable circuits than Unqomp \cite{unqomp} and can be decided in linear time, if the unitary matrix is available.
A gate is called permeable on qubit $i$, if it commutes with the $Z$ operator on this qubit.

\begin{align}
   \text{U} \text{ is permeable on qubit i} \Leftrightarrow \text{U} \text{Z}_i = \text{Z}_i \text{U}
\end{align}

This implies that any controlled gate is permeable on its control qubit because

\begin{align*}
\text{Z}_0 \text{cU} &= \begin{pmatrix} \mathbbm{1} & 0 \\ 0 & -\mathbbm{1} \end{pmatrix} \begin{pmatrix} \mathbbm{1} & 0 \\ 0 & U \end{pmatrix}\\
&= \begin{pmatrix} \mathbbm{1} & 0 \\ 0 & -U \end{pmatrix}\\
&= \text{cU} \text{Z}_0
\end{align*}

However, not every permeable unitary is equal to a controlled gate, for example, $\text{Z}_0 \text{CX}_{01}$.

Why is this property relevant for the Unqomp algorithm? The definining feature of the DAG (directed acyclic graph) representation of Unqomp, is the fact that multiple ``control edges'' can be connected to the same node. This is due to the commutative property of control knobs:
\begin{center}
\scalebox{1.0}{
\Qcircuit @C=1.0em @R=0.2em @!R {
	 	\nghost{} & \lstick{q_0 :} & \targ & \qw & \qw&&& \qw & \targ & \qw & \\
	 	\nghost{} & \lstick{q_1 :} &  \ctrl{-1} & \ctrl{1} & \qw & \push{\rule{.3em}{0em}=\rule{.3em}{0em}} && \ctrl{1} & \ctrl{-1} & \qw\\
	 	\nghost{} & \lstick{q_2 :} & \qw & \targ & \qw &&& \targ & \qw & \qw \\
}}
\end{center}

The DAG representation of Uncomp no longer contains any information about the order in which controlled gates are applied. It therefore supports a flexible insertion of the inverse ``uncomputation'' gates, since it is not necessary to specify, the concrete position in a sequence of controlled gates. In other words, Unqomp's DAG representation abstracts away equivalence classes of gate sequence permutations based on non-trivial commutation relations.

At this point we need the following theorem, which is proved in the \hyperref[appendix]{appendix}:

\begin{theorem}
\label{commutativity_theorem} Let $U \in U(2^n)$ and $V \in U(2^m)$ be $n$ and $m$ qubit operators, respectively. If $U$ is permeable on its last $p$ qubits and $V$ is permeable on its first $p$ qubits, the two operators commute, if they intersect only on these qubits:

\begin{align}
(U \otimes \mathbbm{1}^{\otimes m-p}) (\mathbbm{1}^{\otimes n-p} \otimes V) =  (\mathbbm{1}^{\otimes n-p} \otimes V) (U \otimes \mathbbm{1}^{\otimes m-p})
\end{align}
\end{theorem}

According to Theorem~\ref{commutativity_theorem}, it is not only control knobs that possess the above non-trivial commutation relation but the same is also true for two general gates, $U, V$ if they are both permeable on $q_1$:

\begin{center}
\scalebox{1.0}{
\Qcircuit @C=1.0em @R=1.0em @!R {
	 	\nghost{{q\_0} :  } & \lstick{{q_0} :  } & \multigate{1}{\mathrm{U}} & \qw & \qw &&& \qw & \multigate{1}{\mathrm{U}} & \qw &\qw \\
	 	\nghost{{q\_1} :  } & \lstick{{q_1} :  } & \ghost{\mathrm{U}} & \multigate{1}{\mathrm{V}} & \qw & \push{\rule{.3em}{0em}=\rule{.3em}{0em}} &&\multigate{1}{\mathrm{V}} & \ghost{\mathrm{U}} & \qw & \qw\\
	 	\nghost{{q\_2} :  } & \lstick{{q_2} :  } & \qw & \ghost{\mathrm{V}} & \qw &&& \ghost{\mathrm{V}} & \qw & \qw & \qw\\
}}
\end{center}

We therefore modify the Unqomp algorithm in such a way that, every time it determines whether a gate is controlled on a certain qubit, we instead return the permeability status on that qubit. This simple modification has proved to provide a uniform way of treating uncomputation of synthesized gates. In addition, it also expanded the class of circuits that can be uncomputed. For example, an important class of synthesized gates, that is permeable but not controlled, is quantum logic synthesis.

As mentioned before, permeability can be determined efficiently. This is due to the fact, that according to Theorem ~\ref{expansion_theorem} (in the appendix), the matrix representation is block diagonal. For instance, if $p = 2$:

\begin{align}
    U = \begin{pmatrix}
         \tilde{U}_0 & 0 & 0 & 0\\
        0 & \tilde{U}_1 & 0 & 0 \\
        0 & 0 & \tilde{U}_2 & 0 \\
        0 & 0 & 0 & \tilde{U}_3
    \end{pmatrix}
\end{align}

Therefore, permeability can be decided by iteratively checking the off-diagonal blocks for non-zero entries.

\section{Summary and Conclusions} \label{Summary_and_Conclusions}

In this paper, we gave a short introduction to why uncomputation is necessary in general and how to perform it. We introduced two ways of implementing and using the state-of-the-art algorithm Unqomp \cite{unqomp} (the \verb|auto_uncompute| decorator and the \verb|uncompute| method of the \verb|QuantumVariable| class) in the Qrisp high-level programming framework. Moreover, we gave a short example of how to deploy these techniques, in order to have an even more elegant formulation of solving quadratic equations using Grover's algorithm~\cite{grover} than in our previous article about Qrisp \cite{Qrisp}.
Finally, we elaborated on our extension of the Unqomp algorithm, which supports the uncomputation of more general quantum circuits an efficient way of deciding the necessary properties (permeability, qfree-ness) required for it to work.

\bibliographystyle{plain} 
\bibliography{sources.bib}

\section*{Appendix: Proof of Theorem 1}
\label{appendix}
This appendix contains the proofs for theorem \ref{commutativity_theorem} from section \ref{Improving_Unqomp_for_the_Needs of_Qrisp}. In order to derive the arguments for theorem \ref{commutativity_theorem}, we first need the following theorem.

\begin{theorem}
    \label{expansion_theorem}
    Let $U \in U(2^n)$ be an $n$-qubit operator. If $U$ is permeable on the first $p$ qubits, there are operators $\tilde{U}_0, \tilde{U}_1 .. \tilde{U}_{2^p -1}$ such that:
    \begin{align}
        \label{expansion}
        U = \sum_{i = 0}^{2^p-1} \ket{i}\bra{i} \otimes \tilde{U}_i.
    \end{align}
\end{theorem}

\begin{proof}
    We will treat the case $p = 1$ first and generalize via induction afterwards.\\
    We start by inserting identity operators $\mathbbm{1} = \sum_{i = 0} \ket{i}\bra{i}$:
    \begin{align}
        U = &\mathbbm{1} U \mathbbm{1} \\
        =& \sum_{i,j = 0}^1 \ket{i}\bra{i} U \ket{j}\bra{j}\\
        =& \sum_{i,j = 0}^1 \ket{i} \bra{j} \otimes \hat{U}_{ij}.
    \end{align}
    where $\hat{U}_{ij} = \bra{i} U \ket{j}$.
    Due to the permeability condition, we have
    \begin{align}
        0 =& Z_0 U - U Z_0\\
        = & \left( \sum_{k =0}^1 (-1)^k \ket{k}\bra{k} \otimes \mathbbm{1}^{\otimes{n-1}}\right) \left(\sum_{i,j = 0}^1 \ket{i} \bra{j} \otimes \hat{U}_{ij}\right)\\
        -& \left( \sum_{i,j = 0}^1 \ket{i} \bra{j} \otimes \hat{U}_{ij}\right) \left(\sum_{k = 0}^1 (-1)^k \ket{k}\bra{k} \otimes \mathbbm{1}^{\otimes{n-1}}\right)\\
        = & \sum_{i,j,k = 0}^1 (-1)^k \left(\ket{k} \langle k | i \rangle \bra{j} \otimes \hat{U}_{ij} - \ket{i} \langle j | k \rangle \bra{k} \otimes \hat{U}_{ij}\right)\\
        = & \sum_{i,j,k = 0}^1 (-1)^k \left(\ket{k} \langle k | i \rangle \bra{j} - \ket{i} \langle j | k \rangle \bra{k}\right) \otimes \hat{U}_{ij}\\
        = & \sum_{i,j = 0}^1 \left((-1)^i \ket{i} \bra{j} - (-1)^j \ket{i} \bra{j}\right) \otimes \hat{U}_{ij}.
    \end{align}
    From this form, we see that the index constellations, where $i = j$ cancel out. We end up with
    \begin{align}
        0 = 2 (\ket{0}\bra{1} \otimes \hat{U}_{01} - \ket{1}\bra{0} \otimes \hat{U}_{10}).
    \end{align}
    Since both summands act on disjoint subspaces, we conclude
    \begin{align}
        \hat{U}_{01} = 0 = \hat{U}_{10}.
    \end{align}
    Finally, we set
    \begin{align}
        \tilde{U}_0 = \hat{U}_{00}\\
        \tilde{U}_1 = \hat{U}_{11}\\
    \end{align}
    yielding the claim for p = 1.
    To complete the proof we give the induction step, that is we will proof the claim for $p = p_0 + 1$ under the assumption that it is true for $p = p_0$:
    Since $U$ is permeable on qubit $p_0 + 1$, we have
    \begin{align}
        0 =& Z_{p_0 + 1} U - U Z_{p_0 + 1}
    \end{align}
    As the claim is true for $p = p_0$, we insert
    \begin{align}
        \label{small_expansion}
        U = \sum_{i = 0}^{2^{p_0} - 1} \ket{i}\bra{i} \otimes \tilde{U}_i
    \end{align}
    yielding
    \begin{align}
        0 = \sum_{i = 0}^{2^{p_0} - 1} \ket{i}\bra{i} \otimes (Z_{p_0 +1} \tilde{U}_i - Z_{p_0 +1} \tilde{U}_i )
    \end{align}
    Since each of the summand operators act on disjoint subspaces, we conclude
    \begin{align}
        0 =  Z_{p_0 +1} \tilde{U}_i - Z_{p_0 +1} \tilde{U}_i
    \end{align}
    Implying
    \begin{align}
        \tilde{U}_i = \sum_{j = 0}^{1} \ket{j}\bra{j} \otimes (\tilde{U}_i)_j.
    \end{align}
    We insert this form into into eq. \ref{small_expansion} to retrieve the claim for $p = p_0 + 1$:
    \begin{align}
        U = \sum_{i = 0}^{2^{p_0+1} - 1} \ket{i}\bra{i} \otimes \tilde{U}_i
    \end{align}
    
\begin{flushright}
$\square$\\
\end{flushright}    
\end{proof}

Having proved the above theorem, the next step is to employ it in argumentation for the validity
of Theorem \ref{commutativity_theorem}.

\begin{proof}
According to Theorem \ref{expansion_theorem} we can write
\begin{align}
    (U \otimes \mathbbm{1}^{\otimes m-p}) = \sum_{i = 0}^{2^p-1} \tilde{U}_i \otimes \ket{i}\bra{i} \otimes  \mathbbm{1}^{\otimes m-p}\\
    (\mathbbm{1}^{\otimes n-p} \otimes V) = \sum_{j = 0}^{2^p-1} \mathbbm{1}^{\otimes n-p} \otimes \ket{j}\bra{j} \otimes \tilde{V}_j
\end{align}
Multiplying these operators gives
\begin{align}
    &(U \otimes \mathbbm{1}^{\otimes m-p}) (\mathbbm{1}^{\otimes n-p} \otimes V)\\
    = &\left( \sum_{i = 0}^{2^p-1} \tilde{U}_i \otimes \ket{i}\bra{i} \otimes \mathbbm{1}^{\otimes m-p} \right)\left(\sum_{j = 0}^{2^p-1} \mathbbm{1}^{\otimes n-p} \otimes \ket{j}\bra{j} \otimes \tilde{V}_j\right)\\
    = &\sum_{i,j = 0}^{2^p-1}  \tilde{U}_i \otimes \ket{i} \langle i|j \rangle  \bra{j} \otimes \tilde{V}_j\\
    = &\sum_{i = 0}^{2^p-1}  \tilde{U}_i \otimes \ket{i} \bra{i} \otimes \tilde{V}_i
\end{align}
Multiplication in reverse order yields the same result:
\begin{align}
    &(\mathbbm{1}^{\otimes n-p} \otimes V) (U \otimes \mathbbm{1}^{\otimes m-p})\\
    = & \left(\sum_{j = 0}^{2^p-1} \mathbbm{1}^{\otimes n-p} \otimes  \ket{j}\bra{j} \otimes 
 \tilde{V}_j\right) \left(\sum_{i = 0}^{2^p-1} \tilde{U}_i \otimes  \ket{i}\bra{i} \otimes 
 \mathbbm{1}^{\otimes m-p}\right) \\
    = &\sum_{i,j = 0}^{2^p-1} \tilde{U}_i \otimes  \ket{j} \langle j|i \rangle  \bra{i} \otimes \tilde{V}_j\\
    = &\sum_{i = 0}^{2^p-1}  \tilde{U}_i \otimes  \ket{i} \bra{i} \otimes \tilde{V}_i
\end{align}
\begin{flushright}
$\square$\\
\end{flushright}
\end{proof}

\section*{Acknowledgement:} This work was funded by the Federal Ministry for Economic Affairs and Climate Action (German: Bundesministerium für Wirtschaft und Klimaschutz) under the funding number Qompiler project. The authors are responsible for the content of this publication.

\end{document}